\documentclass[conference,a4paper]{IEEEtran}
\pdfoutput=1

\usepackage[utf8]{inputenc}
\usepackage[innermargin=0.545in,outermargin=0.545in,top=0.545in,bottom=1in]{geometry}
\usepackage[english]{babel}
\usepackage[T1]{fontenc}
\usepackage{epsfig}
\usepackage{amsmath, amssymb, amsbsy}
\usepackage{mathdots}
\usepackage{xspace}
\usepackage[noend]{algpseudocode}
\usepackage{algorithm}
\usepackage{algorithmicx}
\usepackage{color}
\usepackage{cite}
\usepackage{booktabs}
\usepackage{verbatim}
\usepackage[colorlinks=true,citecolor=blue,linkcolor=blue,urlcolor=blue]{hyperref}
\usepackage{url}
\usepackage{lipsum}
\usepackage{tikz}
\usepackage{enumitem}
\usetikzlibrary{arrows,matrix,positioning}
\usepackage{colortbl}
\usepackage{flushend}

\newtheorem{theorem}{Theorem}
\newtheorem{lemma}[theorem]{Lemma}
\newtheorem{corollary}[theorem]{Corollary}

\newtheorem{remark}[theorem]{Remark}
\newtheorem{definition}{Definition}

\usepackage[final,tracking=true,kerning=true,spacing=true,factor=1100,stretch=10,shrink=20]{microtype}

\algrenewcommand\alglinenumber[1]{{\scriptsize#1}}
\algrenewcommand\algorithmicrequire{\textbf{Input:}}
\algrenewcommand\algorithmicensure{\textbf{Output:}}

\usepackage{varioref}        % for \vref{:...} gives "As 1.2 on page 4"
\usepackage{fancyref}

\makeatletter
\def\mkfancyprefix#1#2{%
\expandafter\def\csname fancyref#1labelprefix\endcsname{#1}%
\begingroup\def\x{\endgroup\frefformat{plain}}%
    \expandafter\x\csname fancyref#1labelprefix\endcsname
    {\MakeLowercase{#2}\fancyrefdefaultspacing##1}%
\begingroup\def\x{\endgroup\Frefformat{plain}}%
    \expandafter\x\csname fancyref#1labelprefix\endcsname
    {#2\fancyrefdefaultspacing##1}%
\begingroup\def\x{\endgroup\frefformat{vario}}%
    \expandafter\x\csname fancyref#1labelprefix\endcsname
    {\MakeLowercase{#2}\fancyrefdefaultspacing##1##3}%
\begingroup\def\x{\endgroup\Frefformat{vario}}%
    \expandafter\x\csname fancyref#1labelprefix\endcsname
    {#2\fancyrefdefaultspacing##1##3}%
}
\makeatother
\fancyrefchangeprefix{\fancyrefeqlabelprefix}{eqn}
\mkfancyprefix{ssec}{Section}
\mkfancyprefix{tbl}{Table}
\mkfancyprefix{thm}{Theorem}
\mkfancyprefix{lem}{Lemma}
\mkfancyprefix{cor}{Corollary}
\mkfancyprefix{prop}{Proposition}
\mkfancyprefix{prob}{Problem}
\mkfancyprefix{alg}{Algorithm}
\mkfancyprefix{inv}{Invariant}
\mkfancyprefix{ex}{Example}
\mkfancyprefix{line}{Line}
\mkfancyprefix{def}{Definition}
\mkfancyprefix{itm}{Item}
\mkfancyprefix{app}{Appendix}
\mkfancyprefix{rem}{Remark}
\newcommand{\cref}[1]{\Fref{#1}}

\def\ve#1{{\mathchoice{\mbox{\boldmath$\displaystyle #1$}}%
              {\mbox{\boldmath$\textstyle #1$}}%
              {\mbox{\boldmath$\scriptstyle #1$}}%
              {\mbox{\boldmath$\scriptscriptstyle #1$}}}}

\definecolor{sunset}{rgb}{1,0.5,.05}
\definecolor{marine}{rgb}{0,0,.7}
\definecolor{navy}{rgb}{0,0,.5}
\definecolor{forest}{rgb}{0,.6,0}
\definecolor{brown}{rgb}{0.59, 0.29, 0.0}

\newcommand{\jsrn}[1]{{\color{sunset}[/jsrn: #1]}}
\newcommand{\pu}[1]{{\color{marine}[/pu: #1]}}
\newcommand{\pabe}[1]{{\color{forest}[/pabe: #1]}}
\newcommand{\alternative}[1]{{\color{brown}[/alternative: #1]}}
 \renewcommand{\jsrn}[1]{}
 \renewcommand{\pu}[1]{}
 \renewcommand{\pabe}[1]{}
 \renewcommand{\alternative}[1]{}

\newcommand{\mo}{{-1}}

\newcommand{\F}{\mathbb{F}}

\newcommand{\ZZ}{\mathbb{Z}}

\renewcommand{\v}{\ensuremath{\ve{v}}}

\newcommand{\NN}{\mathbb{N}}

\AtBeginDocument{}

\newcommand{\Fq}{\mathbb{F}_q}
\newcommand{\Fs}{\mathbb{F}_s}

\newcommand{\evpolys}{\mathcal{V}_{k,t,h,\eta}}

\newcommand{\Code}{\mathcal{C}}
\newcommand{\TRS}[3]{\Code_{#2}^{#3}(#1)}

\newcommand{\Cstar}{\TRS{\alphaVec,\eta}{k}{*}}
\newcommand{\Cplus}{\TRS{\eta,\ast}{k}{+}}
\newcommand{\Csingle}{\TRS{\alphaVec,t,h,\eta}{k}{}}

\newcommand{\twisted}{$(k,t,h,\eta)$-twisted }

\newcommand{\CGRS}[1]{\mathcal{C}_{#1}^\mathrm{GRS}}

\newcommand{\Iset}{\mathcal{I}}
\newcommand{\Af}{\Iset}
\newcommand{\Zmatrix}{\ve 0}

\newcommand{\A}{\ve A}
\newcommand{\G}{\ve G}

\newcommand{\I}{\ve I}

\newcommand{\x}{\ve x}

\newcommand{\fpolyi}[1]{f^{(#1)}}

\newcommand{\Atilde}{\ve{\tilde{A}}}

\newcommand{\ev}[2]{\mathrm{ev}_{#2}(#1)}

\newcommand{\transpose}{^\mathrm{T}}

\newcommand{\alphaVec}{{\ve \alpha}}

\newcommand{\g}{\ve g}
\newcommand{\etaSet}{\mathcal{H}}

\newcommand{\Startw}{$(*)$-Twisted }
\newcommand{\startw}{$(*)$-twisted }

\newcommand{\plustw}{$(+)$-twisted }

\newcommand\Osoft{O^{\scriptscriptstyle \sim}\!}

\definecolor{myred}{rgb}{0.7,0,0}

\begin{document}

\title{Twisted Reed--Solomon Codes}
\author{\IEEEauthorblockN{Peter Beelen$^{1}$, Sven Puchinger$^{2}$, and Johan Rosenkilde n\'e Nielsen$^{1}$}\\
\IEEEauthorblockA{
$^1$Department of Applied Mathematics \& Computer Science, Technical University of Denmark, Lyngby, Denmark\\
$^2$Institute of Communications Engineering, Ulm University, Ulm, Germany\\
Email: \emph{pabe@dtu.dk, sven.puchinger@uni-ulm.de, jsrn@jsrn.dk}
}}

\maketitle

\begin{abstract}
We present a new general construction of MDS codes over a finite field $\Fq$.
We describe two explicit subclasses which contain new MDS codes of length at least $q/2$ for all values of $q \ge 11$.
Moreover, we show that most of the new codes are not equivalent to a Reed--Solomon code.
\end{abstract}


\section{Introduction}

\noindent
A \emph{maximum distance separable} (MDS) code $\Code(n,k,d)$ of length $n$, dimension $k$, and minimum distance $d$ is a linear code attaining the Singleton bound, i.e., $d=n-k+1$ \cite{singleton1964maximum}.
The most prominent MDS codes are \emph{generalized Reed--Solomon} (GRS) codes \cite{reed1960polynomial}.
However, there are many other known constructions for MDS codes, e.g., based on the equivalent problem of finding $n$-arcs in projective geometry \cite{macwilliams1977theory}, circulant matrices \cite{roth_construction_1989}, Hankel matrices \cite{roth1985generator}, or extending GRS codes.

Recently, Sheekey \cite{sheekey_new_2015} introduced a new class of \emph{maximum rank distance} codes---which are MDS in terms of the \emph{rank metric}.
These codes, \emph{Twisted Gabidulin} codes, were shown to be not equivalent to Gabidulin codes (the rank-metric analogues of Reed--Solomon codes).

In this paper, we introduce a new construction of Hamming-metric MDS codes, inspired by~\cite{sheekey_new_2015}.
The idea is to evaluate polynomials of the form
\[
  f(x) = a_0 + a_1 x + \ldots + a_{k-1} x^{k-1} + \eta a_h x^{k-1+t} ,
\]
for some $0 \leq h < k$ and $t < n-k$.
By a prudent choice of $t, h$ and $\eta$, we ensure that any such polynomial has at most $k-1$ zeroes among the evaluation points, even though their degree is larger than $k-1$.
This is enough to ensure that the resulting code is MDS.

We single out two explicit subclasses of this construction where the MDS property can be a priori ensured.
These contain codes of length up to roughly $q/2$ for any $q$, and we show that for $q \geq 11$, they contain non-GRS MDS codes.

The results we obtain are somewhat reminiscent of results in \cite{roth_construction_1989}, where non-GRS MDS codes were constructed of length roughly $q/2$ for even $q$.
However, our construction is very different, and for small values of $q$ we verified using a computer that our construction produces codes inequivalent to the ones mentioned in \cite{roth_construction_1989}.
More importantly, our construction also gives a very simple way to produce non-GRS MDS codes of length at least $q/2$ if $q$ is odd and $q \ge 11$. 

Besides adding new codes to the family of known MDS codes, the new code class might be interesting for code-based cryptography.
As future work, we will analyze whether our codes or their subfield subcodes are suitable for this purpose.

\section{Preliminaries}

\noindent
In this section, we recall several definitions and known results for future use in the paper. We denote by $\Fq$ the finite field with $q$ elements.

\begin{definition}[\!\!\cite{roth2006introduction}]
Let $\alpha_1,\dots,\alpha_n \in \Fq \cup \{\infty\}$ be distinct, $k<n$, and $v_1,\dots,v_n \in \Fq^*$. The corresponding \emph{generalized Reed--Solomon (GRS) code} is defined by
\begin{align*}
\CGRS{n,k} = \left\{ [v_1 f(\alpha_1), \dots, v_n f(\alpha_n)] : \, f \in \Fq[x], \, \deg f <k \right\}.
\end{align*}
In this setting, for a polynomial $f$ of degree $\deg f < k$, the quantity $f(\infty)$ is defined as $a_{k-1}$, the coefficient of $x^{k-1}$ in the polynomial $f$.
\end{definition}

In case $v_i=1$ for all $i$, the code is called a \emph{Reed--Solomon (RS) code}. Any non-zero evaluation polynomial $f$ is of degree $\deg f <k$ and hence has at most $k-1$ roots among the evaluation points $\alpha_1,\dots,\alpha_n$ in $\Fq$. If $f(\infty)$ ``evaluates'' to zero, this just means that $\deg f < k-1$ and hence $f$ has at most $k-2$ roots among the remaining evaluation points. This proves that a GRS code is MDS.

In this article, we will construct MDS codes of length $n$ and dimension $k$ using spaces of polynomials that may contain elements of degree $\deg f \geq k$. To define our codes, we will use the following map.

\begin{definition}\label{def:eval_map}
Let $\mathcal{V} \subset \Fq[X]$ be a $k$-dimensional $\Fq$-linear subspace. Let $\alpha_1,\dots,\alpha_n \in \Fq \cup \{\infty\}$ be distinct and write $\alphaVec = [\alpha_1,\dots,\alpha_n]$. We call $\alpha_1,\dots,\alpha_n$ the \emph{evaluation points}. Then we define the \emph{evaluation map} of $\alphaVec$ on $\mathcal{V}$ by
\begin{align*}
\ev{\cdot}{\alphaVec} : \mathcal{V} \to \Fq^n, \quad
f \mapsto [f(\alpha_1),\dots,f(\alpha_n)].
\end{align*}
Here $f(\infty)$ is defined as $a_{\ell}$, the coefficient of $x^{\ell}$ in the polynomial $f \in \mathcal{V}$, where $\ell:=\max \deg\{f : f \in \mathcal{V}\}.$
\end{definition}

The evaluation map above is $\Fq$-linear. This means in particular that if for a given $\mathcal{V}$ the evaluation map is injective, then the code $\ev{\mathcal V}{\alphaVec} \subset \Fq^n$ will be an $\Fq$-linear code of length $n$ and dimension $k$.
Injectivity is immediate if $\ell < n$.
If $\mathcal{V}$ consists of all polynomials of degree strictly less than $k$, the resulting code is an RS code.
For other choices of $\mathcal{V}$, the resulting code might still be \emph{equivalent} to an RS code; we use the following notion of code equivalence.

\begin{definition}
Let $\Code_1,\Code_2$ be $\Fq$-linear $[n,k]$ codes. We say that $\Code_1$ and $\Code_2$ are \emph{equivalent} if there is a permutation $\pi \in S_n$ and $\v := [v_1,\dots,v_n] \in \left(\Fq^*\right)^n$ such that $\Code_2 = \varphi_{\pi,\v}\left( \Code_1 \right)$ where $\varphi_{\pi, \v}$ is the Hamming-metric isometry
\begin{align*}
\varphi_{\pi,\v} : \Fq^n \to \Fq^n, \quad
[c_1,\dots,c_n] \mapsto [v_1 c_{\pi(1)},\dots,v_n c_{\pi(n)}].
\end{align*}
\end{definition}

It is clear from the above definition that a code is a GRS code if and only if it is equivalent to an RS code.
The following well-known theorem provides an effective tool to decide whether a code is equivalent to an RS code.

\begin{theorem}[\!\!\cite{roth1985generator,roth1989mds}]\label{thm:GRS_characterization}
A linear code with generator matrix $\G = [\I \mid \A]$ is a GRS code if and only if
\begin{enumerate}[label=(\roman*)]
\item\label{itm:GRS_characterization_1} All entries of $\A$ are non-zero.
\item\label{itm:GRS_characterization_2} All $2 \times 2$ minors of $\Atilde$ are non-zero, and
\item\label{itm:GRS_characterization_3} All $3 \times 3$ minors of $\Atilde$ are zero,
\end{enumerate}
where $\Atilde \in \Fq^{k \times n-k}$ is given by $\tilde{A}_{ij} = A_{ij}^{-1}$.
\end{theorem}

Note that any MDS code has a generator matrix of the form $\G = [\I \mid \A]$ and that items \ref{itm:GRS_characterization_1} and \ref{itm:GRS_characterization_2} above are satisfied for this $A$. Hence the difference between GRS and non-GRS MDS codes will only become apparent using item \ref{itm:GRS_characterization_3}. If $k<3$ or $n-k<3$, the matrix $A$ has no $3 \times 3$ minors, so the following corollary holds.
\begin{corollary}\label{cor:min_k_n-k_>_2}
Suppose that $k<3$ or $n-k<3$. Any MDS code of length $n$ and dimension $k$ is equivalent to an RS code.
\end{corollary}

We finish this section by quoting results on $t$-sum generators in abelian groups from \cite{roth_construction_1989, roth_t-sum_1992}. We will apply these results to the abelian groups $(\Fq^*,\cdot)$ and $(\Fq,+)$ to analyse several instances of our code construction in the coming sections.

\begin{definition}[\!\!\cite{roth_t-sum_1992}]
Let $(A,\oplus)$ be a finite abelian group and $k \in \NN$. A subset $S \subset A$ is called a \emph{$k$-sum generator} of $A$ if for all $a \in A$, there are distinct $s_1,\dots,s_k \in S$ such that $a=\bigoplus_{i=1}^{k} s_i$. We denote by $M(k,A)$ the smallest integer such that any $S \subset A$ with $|S|>M(k,A)$ is a $k$-sum generator of $A$.
\end{definition}

\begin{lemma}[\!\!{\cite[Theorem~3.1]{roth_t-sum_1992}}]\label{lem:tsum_roth_theorem_even_order}
Let $|A|= 2r$ for some $r \geq 6$. Then, for any $3 \leq k \leq r-2$,
\begin{align*}
M(k,A) = r,
\end{align*}
except if $A \in \{\ZZ_2^m,\ZZ_4 \times \ZZ_2^{m-1}\}$ for some $m > 1$ and $k \in \{3,r-2\}$ in which case
\begin{align*}
M(k,A) = r+1.
\end{align*}
\end{lemma}
This lemma suffices for our purposes. See \cite{roth_construction_1989, roth_t-sum_1992} for more information in case $|A|$ is odd.

\section{Twisted Reed--Solomon Codes}\label{sec:CTRS_one_twist}

\noindent
In this section, we present our new code construction.
Similar to RS codes, we evaluate polynomials whose first $k$ coefficients we can choose arbitrarily.
The difference is that we allow another monomial of degree larger than $k-1$ to occur in the polynomials as well.
We define a set of evaluation polynomials as follows.

\begin{definition}\label{def:twisted_polynomials}
Let $k,t,h \in \NN$ such that $0 \le h<k\leq q$ and let $\eta \in \Fq \backslash \{0\}$. Then, we define the set of \emph{\twisted polynomials} by
\begin{align*}
\evpolys = \left\{ f = \sum_{i=0}^{k-1} a_i x^i + \eta a_h x^{k-1+t} : a_i \in \Fq \right\},
\end{align*}
where we call $h$ the \emph{hook} and $t$ the \emph{twist}.
\end{definition}

Note that $\evpolys \subseteq \Fq^n$ is a $k$-dimensional $\Fq$-linear subspace. Using the evaluation map from Definition \ref{def:eval_map} for $\mathcal{V}=\evpolys$, we obtain the codes that we will study in this article.

\begin{definition}\label{def:CTRS_one_twist}
Let $\alpha_1,\dots,\alpha_n \in \Fq \cup \{\infty\}$ be distinct and write $\alphaVec = [\alpha_1,\dots,\alpha_n]$.
Let $k,t,h,\eta$ be chosen as in Definition~\ref{def:twisted_polynomials} such that $k<n$ and $t \leq n-k$. Then, the corresponding \emph{twisted Reed--Solomon code} of length $n$ and dimension $k$ is given by
\begin{align*}
\Csingle = \ev{\evpolys}{\alphaVec} \subseteq \Fq^n.
\end{align*}
\end{definition}

For brevity, we will use the phrase \emph{twisted codes} rather than \emph{twisted Reed--Solomon codes} from now on.
Note that $\Csingle$ indeed has dimension $k$ since the evaluation map is injective: any polynomial $f \in \evpolys$ satisfies $\deg f \le k-1+t < n$. In principle $\eta$ could be $0$ in the above definition, but in that case we simply obtain RS codes.

\section{MDS Twisted Codes}
\label{sec:MDS_TRS_Codes}

In general, twisted codes are not MDS for all parameters $\alphaVec,k,t,h,\eta$. However, in this section we will describe several classes of twisted codes that are always MDS.

\subsection{\startw Codes}
\label{ssec:Max_Length_MDS_Primitively Twisted}

\noindent
If $(t,h)=(1,0)$, it is possible to give a succinct condition on when the code $\TRS{1,0,\eta,\alphaVec}{k}{}$ is MDS. More precisely, we have the following:
\begin{lemma}\label{lem:cond_MDS_very_simple_twist}
Let $k<n$, $\alpha_1,\dots,\alpha_n \in \Fq$ distinct and $\eta \in \Fq$. Then the twisted code $\TRS{\alphaVec,1,0,\eta}{k}{}$ is MDS if and only if
\begin{align}
\eta (-1)^k \prod_{i \in \Iset} \alpha_i \neq 1 \quad \forall \Iset \subseteq \{1,\dots,n\} \textrm{ s.t. } |\Iset|=k. \label{eq:cond_MDS_very_simple_twist}
\end{align}
\end{lemma}

\begin{proof}
The code is MDS if and only if the only polynomial of the required form which has $k$ roots among the $\alpha_i$ is the zero polynomial. Let $f = \sum_{i=0}^{k-1} a_i x^i + \eta a_0 x^k$ be such a polynomial. In the first place $a_0 \neq 0$, since otherwise $\deg f < k$, making it impossible that $f$ has $k$ roots among the $\alpha_i$. If there is a subset $\Iset \subseteq \{1,\dots,n\}$ with $|\Iset|=k$ and $f(\alpha_i) = 0$ for all $i \in \Iset$, we can write $f = \eta a_0 \prod_{i \in \Iset}(x-\alpha_i)$ and considering the constant term it follows that
\begin{align}
1 = \eta (-1)^k \prod_{i \in \Iset} \alpha_i. \label{eq:intermediate_step_MDS_very_simple_twist}
\end{align}
If Condition \eqref{eq:cond_MDS_very_simple_twist} is satisfied, no such $f$ can exist.
Conversely, if there is an $\Iset$ such that $\eta (-1)^k \prod_{i \in \Iset} \alpha_i = 1$, then $\eta\neq 0$ and $f = \eta \prod_{i \in \Iset}(x-\alpha_i) \in \evpolys$ has $k$ roots among the $\alpha_i$, so the code is not MDS.
\end{proof}

This leads to our first explicit subclass of twisted codes which are MDS\footnote{%
  This class can be seen as the Hamming-metric analog of Twisted Gabidulin codes \cite{sheekey_new_2015}. However, we use different techniques for analyzing our codes.}:
\begin{definition}\label{def:startwisted}
  $\TRS{\alphaVec,1,0,\eta}{k}{}$ is a \emph{\startw code} if the elements of $\alphaVec$ are a subset of $G \cup \{ 0 \}$, for $G$ a proper subgroup of $(\F_q^*, \cdot)$, and if $(-1)^k \eta^\mo \in \F_q^* \setminus G$.
  We write $\Cstar := \TRS{\alphaVec,1,0,\eta}{k}{}$.
\end{definition}

\begin{theorem}\label{thm:MDS_property_simple_twist}
Any \startw code is an MDS code.
\end{theorem}
\begin{proof}
If a \startw code $\Cstar$ is not MDS, then Lemma \ref{lem:cond_MDS_very_simple_twist} implies that there exists $\Iset \subseteq \{1,\dots,n\}$ such that $(-1)^k\eta\prod_{i \in \Iset} \alpha_i =1$. Since the $\alpha_i$ are contained in a subgroup $G$ of $\Fq^*$, we have $\prod_{i \in \Iset} \alpha_i \in G$, implying that $(-1)^k\eta \in G$ as well. This gives a contradiction.
\end{proof}

\begin{corollary}\label{cor:MDS_primitively-twisted}
Let $\Fq$ be a finite field and let $p$ be a prime divisor of $q-1$. Then there exists a \startw code of length
\begin{align*}
n = \tfrac{q-1}{p}+1.
\end{align*}
In particular, if $q$ is odd, \startw codes can have length $n = \frac{q+1}{2}$.
\end{corollary}
\begin{proof}
The maximum cardinality of a proper subgroup $G$ of $\Fq^*$ is $(q-1)/p$. Now let $\alphaVec$ be $G \cup \{0\}$ in some order in Definition \ref{def:startwisted}.
\end{proof}

For odd $q$ \startw codes can therefore be rather long.
But the choice of $\alphaVec$ is very limited in \cref{def:startwisted}, and perhaps longer codes could be constructed with $(t,h) = (1,0)$.
The following answers this negatively, by using $k$-sum generators.
\begin{lemma}\label{lem:k_sum_gen_star_MDS}
Let $\eta \in \F_q^*$ and let $S := \{\alpha_1,\dots,\alpha_n\} \subseteq \Fq^*$ be a $k$-sum generator of $(\Fq^*,\cdot )$.
Then the twisted code $\TRS{\alphaVec,1,0,\eta}{k}{}$ is not MDS.
\end{lemma}
\begin{proof}
By the definition of a $k$-sum generator, there is an index set $\Iset \subseteq \{1,\dots,n\}$ with $|\Iset| = k$ such that $\prod_{i \in \Iset} \alpha_i = (-1)^k \eta^{-1}$.
Hence $\TRS{\alphaVec,1,0,\eta}{k}{}$ is not MDS by \cref{lem:cond_MDS_very_simple_twist}.
\end{proof}

\begin{theorem}\label{thm:inverse_mds_statement_very_simple_twist}
Let $\Fq$ be a finite field, with $q$ odd. Further let $\eta \in \Fq^*$. Then, for any $3 \leq k \leq \tfrac{q-1}{2}-2$, the length $n$ of a twisted code $\TRS{\alphaVec,1,0,\eta}{k}{}$ which is MDS, satisfies
$
n \leq \frac{q+1}{2}.
$
\end{theorem}

\begin{proof}
We know that $(\Fq^*,\cdot)$ is a cyclic abelian group of order $|\Fq^*|=q-1$, which is even since $q$ is a power of an odd prime. Hence \cref{lem:tsum_roth_theorem_even_order} implies
\begin{align*}
M(k,\Fq^*) = \tfrac{q-1}{2}.
\end{align*}
Now suppose $n>\tfrac{q+1}{2}$ and $\alphaVec=[\alpha_1,\dots,\alpha_n]$. Then, $S := \{\alpha_1,\dots,\alpha_n\} \cap \Fq^*$ has cardinality
\begin{align*}
|S| \geq n-1 > \tfrac{q+1}{2}-1 = \tfrac{q-1}{2} = M(k,\Fq^*)
\end{align*}
and is therefore a $k$-sum generator of $\Fq^*$.
By \cref{lem:k_sum_gen_star_MDS}, the code $\TRS{\alphaVec,1,0,\eta}{k}{}$ is not MDS for any $\eta \neq 0$.
\end{proof}

\begin{remark}
For even $q>4$, the \startw codes cannot attain length $\lfloor(q+1)/2\rfloor$. However, we determined by computer search that for e.g.~$q=16$, there are many $[9,k]$ twisted codes with $(t,h)=(1,0)$ for other choices of $\alphaVec$ and $\eta$, for $k = 3,4,5$.
See also \cref{sec:Computer_Searches}.
\end{remark}

\subsection{\plustw Codes}
\label{ssec:Max_Length_MDS_Additive-Primitively_Twisted}

\noindent
While the results in the previous subsection were based on properties of the multiplicative group $(\Fq^*,\cdot)$, it is also possible to use the structure of the additive group $(\Fq,+)$. This structure arises when considering the case $(t,h)=(1,k-1)$.
We have the following analogue of \cref{lem:cond_MDS_very_simple_twist}.
\begin{lemma}\label{lem:cond_MDS_primitively_twisted_h=k-1}
Let $k<n\leq q$, $\alpha_1,\dots,\alpha_n \in \Fq$ distinct and $\eta \in \Fq$. Then the twisted code $\TRS{\alphaVec,1,k-1,\eta}{k}{}$ is MDS if and only if
\begin{align}
\eta \sum_{i \in \Iset} \alpha_i \neq -1 \quad \forall \Iset \subseteq \{1,\dots,n\} \textrm{ s.t. } |\Iset|=k. \label{eq:cond_MDS_primitively_twisted_h=k-1}
\end{align}
\end{lemma}
\begin{proof}
  The code is MDS if and only if the only polynomial
  \[
    f = a_0 + \ldots + a_{k-1}x^{k-1} + \eta a_{k-1} x^k \in \mathcal{V}_{k,1,k-1,\eta}
  \]
  having $k$ roots among the $\alpha_i$ is the zero polynomial.
If $f \neq 0$ is such a polynomial, there is a subset $\Iset \subseteq \{1,\dots,n\}$ with $|\Iset|=k$ and $f(\alpha_i) = 0$ for all $i \in \Iset$.
Writing $f = \eta a_{k-1} \prod_{i \in \Iset}(x-\alpha_i)$, with $a_{k-1} \neq 0$, we obtain a contradiction by considering the coefficient of $x^{k-1}$ on both sides, since
\begin{align*}
  a_{k-1} = -\eta a_{k-1} \sum_{i \in \Iset} \alpha_i
  \iff
  \eta \sum_{i \in \Iset} \alpha_i = -1
\end{align*}
Conversely, if there is an $\Iset$ such that $\eta \sum_{i \in \Iset}(-\alpha_i) = 1$, then $f = \eta \prod_{i \in \Iset}(x-\alpha_i)$ is a polynomial of the appropriate form having $k$ roots among the $\alpha_i$, so the code is not MDS.
\end{proof}

As in the previous subsection, this naturally gives rise to a subclass of twisted codes.
\begin{definition}\label{def:plustwisted}
  $\TRS{\alphaVec,1,0,\eta}{k}{}$ is a \emph{\plustw code} if the elements of $\alphaVec$ are a subset of $V \cup \{\infty\}$, for $V$ a proper subgroup of $(\F_q, +)$, and if $\eta^\mo \in \F_q\setminus V$.
  We write $\Cplus := \TRS{\alphaVec,1,0,\eta}{k}{}$.
\end{definition}

The analysis of these codes follows that of their multiplicative counterparts from the previous subsection very closely.
In particular we have the following:
\begin{theorem}\label{thm:MDS_property_primitively_twisted_h=k-1}
Any \plustw code is an MDS code.
\end{theorem}
\begin{proof}
We simply apply \cref{lem:cond_MDS_primitively_twisted_h=k-1} instead of Lemma \ref{lem:cond_MDS_very_simple_twist}. Some care must be taken that adding $\infty$ to a set of evaluation points preserves the MDS property. However, if a polynomial $f \in \mathcal{V}_{k,1,k-1,\eta}$ satisfies $f(\infty)=0$, that means its degree is at most $k-2$.
\end{proof}

\begin{corollary}\label{cor:MDS_primitively-twisted_h=k-1}
Let $\Fq$ be a finite field of characteristic $p$. Then there exists a \plustw code of length
\begin{align*}
n = \tfrac{q}{p}+1.
\end{align*}
In particular, if $q$ is even, \plustw codes can have length $n = \frac{q}{2}+1$.
\end{corollary}
\begin{proof}
The maximum cardinality of a proper subgroup $V$ of $\Fq$ is $q/p$. Now set $S=\{\infty\} \cup V$ in Definition \ref{def:plustwisted}. \end{proof}

Lemma \ref{lem:k_sum_gen_star_MDS} and Theorem \ref{thm:inverse_mds_statement_very_simple_twist} have a direct analogue as well. For completeness, we state the results, but since the proofs are extremely similar, we leave these to the reader.
\begin{lemma}
  \label{lem:k_sum_gen_plus_MDS}
  Let $\eta \neq 0$ and $S := \{\alpha_1,\dots,\alpha_n\} \in \Fq$ be a $k$-sum generator of $(\Fq,+)$.
  Then the twisted code $\TRS{\alphaVec,1,k-1,\eta}{k}{}$ is not MDS.
\end{lemma}

\begin{theorem}\label{thm:inverse_mds_statement_very_simple_twist_h=k-1}
Let $\Fq$ be a finite field, with $q$ even. Further let $\eta \in \Fq^*$. Then, for any $3 \leq k \leq \tfrac{q}{2}-2$, the length $n$ of a twisted code $\TRS{\alphaVec,1,k-1,\eta}{k}{}$ which is MDS, satisfies $n \leq \frac{q}{2}+1$ if $3< k < \frac{q}{2}-3$ and $n \leq \frac{q}{2}+2$ if $k \in \{3,\frac{q}{2}-2\}.$
\end{theorem}

\subsection{General Theory of Twisted Codes}
\label{subsec:Max_Length_MDS_General}

For general $t$ and $h$, it is still possible to derive a criterion for a code $\Csingle$ to be MDS. We do so in the following lemma, generalizing Lemmas \ref{lem:cond_MDS_very_simple_twist} and \ref{lem:cond_MDS_primitively_twisted_h=k-1}.

\begin{lemma}\label{lem:MDS_property_one_twist}
Let $\alpha_1,\dots,\alpha_n \in \Fq$ and define for $\Iset \subseteq \{1,\dots,n\}$ the polynomial $\sum_i \sigma_i x^i := \prod_{i \in \Iset} (x-\alpha_i)$, where $\sigma_i := 0$ for $i<0$. The twisted code $\Csingle$ is MDS if and only if the matrix
\begin{align*}
\underset{=: \, \A_\Iset \, \in \Fq^{t \times t}}{\underbrace{
\begin{bmatrix}
\eta^{-1} -\sigma_{h-t+1} & -\sigma_{h-t+2} & \dots  & -\sigma_{h} & \\
\sigma_{k-1} & 1 & & & & \\
\sigma_{k-2} & \sigma_{k-1} & 1 &  & & \\
\sigma_{k-3} & \sigma_{k-2} & \sigma_{k-1} & 1 &  & & \\
\vdots & \ddots & \ddots & \ddots & \ddots & & \\
\sigma_{k-t+1} & \dots &\sigma_{k-3} & \sigma_{k-2} & \sigma_{k-1} & 1\\
\end{bmatrix}
}},
\end{align*}
is regular for all $\Iset \subseteq \{1,\dots,n\}$ such that $|\Iset|=k$.
\end{lemma}

\begin{proof}
Let $f \in \evpolys$ be a polynomial with at least $k$ roots among the $\alpha_i$'s.
Then, there is an index set $\Af \subseteq \{1,\dots,n\}$ with $|\Af|=k$ and $f(\alpha_i) = 0$ for all $i \in \Af$.
We can factor $f$ into $f(x) = g(x) \cdot \sigma(x)$, with
\begin{align*}
g(x) = \sum\limits_{i=0}^{t-1} g_i x^i, \quad
\sigma(x) = \sum\limits_{i=0}^{k} \sigma_i x^i := \prod_{i \in \Af} (x-\alpha_i).
\end{align*}
Note that $\sigma_k=1$.
Since by construction, the coefficients in $f$ to $x^k, \ldots, x^{k+t-2}$ are zero, we obtain the following system of $(t-1)$ equations in the $g_j$'s,
\begin{align}
0 = \sum\limits_{j=0}^{i} g_j \sigma_{i-j} \quad i=k,\dots,k+t-2, \label{eq:one_hook_MDS_first_equations}
\end{align}
where $g_j := 0$ for all $j \notin \{ 0,\ldots, t-1 \}$, and $\sigma_j:=0$ for $j \notin \{0,\ldots,k\}$.
Considering the coefficients of $x^{k-1+t}$ and $x^h$, we obtain
$\eta a_h   = g_{t-1}$ and $a_h  = \sum_{j=0}^{h} g_j \sigma_{h-j}$
and hence
\begin{align}
0 = \eta^{-1} g_{t-1} -\sum\limits_{j=0}^{h} g_j \sigma_{h-j}. \label{eq:one_hook_MDS_fourth_equation}
\end{align}
Equations \eqref{eq:one_hook_MDS_first_equations} and \eqref{eq:one_hook_MDS_fourth_equation} result in a homogeneous system of $t$ equations in $t$ variables $\g := [g_{t-1},g_{t-2},\dots,g_{0}]\transpose$:
\begin{align}
\A_\Iset \cdot \g
=
\Zmatrix \label{eq:MDS_proof_system_one_twist}
\end{align}
The code $\Csingle$ is MDS if and only if the only polynomial $f \in \evpolys$ with at least $k$ roots is the zero polynomial, which holds if and only if the system \eqref{eq:MDS_proof_system_one_twist} has only the zero vector as solution for all choices of index sets $\Iset \subseteq \{1,\dots,n\}$ with $|\Iset|=k$. This implies the claim.
\end{proof}

\begin{remark}
If one includes $\infty$ as evaluation point and $h=k-1$, the above lemma is still true when considering $\Iset$ not containing $\infty$. Indeed if $f(\infty)=0$, then $\deg f < k-1$, which means the MDS property is not affected.
\end{remark}

As we will see in \cref{sec:Computer_Searches} many long MDS codes can be obtained using twisted codes for particular values of $q$. Hence an upper bound like in Theorems \ref{thm:inverse_mds_statement_very_simple_twist} and \ref{thm:inverse_mds_statement_very_simple_twist_h=k-1} does not hold for general $h$ and $t$.
On the other hand, it seems harder to find explicit constructions of such long MDS codes.
We do have the following result.
\begin{theorem}\label{thm:MDS_property_one_twist}
Let $\Fs \subsetneq \Fq$ be a proper subfield and $\alpha_1,\dots,\alpha_n \in \Fs$.
If $\eta \in \Fq \setminus \Fs$, then the twisted code $\Csingle$ is MDS.
\end{theorem}

\begin{proof}
Let $\eta \in \Fq \setminus \Fs$.
Let $\Iset \subseteq \{1,\dots,n\}$ be an index set with $|\Iset|=k$ and $\A_\Iset$ be the corresponding matrix as in \cref{lem:MDS_property_one_twist}.
Using elementary row operations, we can bring $\A_\Iset$ into lower triangular form
with diagonal elements $\eta^{-1} + T,1, \ldots, 1$ for a certain $T \in \Fq$. Using that $\sigma_i \in \F_s$ for all $i$ (since the same holds for all $\alpha_i$), we conclude that in fact $T \in \Fs$.
Since $\eta \not\in \Fs$, the diagonal elements of the triangular form of $\A_\Iset$ are all non-zero, implying that $\A_\Iset$ is regular.
\end{proof}

\section{Non-GRS MDS Twisted Codes}

\noindent
Since GRS codes are always MDS and well studied, we will in this section show that most of the twisted codes are not equivalent to an RS code. The main result is the following theorem.

\begin{theorem}\label{thm:not_GRS_one_twist}
Let $\alpha_1,\dots,\alpha_n \in \Fq$ and $2 < k < n-2$. 
Furthermore, let $\etaSet \subseteq \Fq$ satisfy that the twisted code $\Csingle$ is MDS for all $\eta \in \etaSet$.
Then there are at most $6$ choices of $\eta \in \etaSet$ such that $\Csingle$ is equivalent to an RS code.
\end{theorem}
\begin{proof}
Since $\Csingle$ is an MDS code, it has a generator matrix of the form $\G = [\I \mid \A]$. Equivalently, there exist polynomials $\fpolyi{i} \in \evpolys$ such that for all $i$ and $j$ in $1,\ldots, k$ it holds that $\fpolyi{i}(\alpha_j) = 1$ if $i=j$ and $\fpolyi{i}(\alpha_j) = 0$ if $i \neq j.$ Further for $j>k$ we have $\fpolyi{i}(\alpha_j) = \A_{i,j-k}$. In particular, since $\fpolyi{i} \in \evpolys$, the $(i,j)$th entry of $\A$ is of the form $\A_{i,j}=c_{i,j}+d_{i,j}\eta$ for certain $c_{i,j},d_{i,j} \in \Fq.$

Now we use item \ref{itm:GRS_characterization_3} of Theorem~\ref{thm:GRS_characterization} to derive an upper bound on the number of choices of $\eta$ such that $\Csingle$ is equivalent to an RS code. Let us consider the minor $M$ of the first three rows and columns of $\Atilde$. Then $\Csingle$ is not equivalent to an RS code if $M$ does not vanish. However, since $\Atilde_{i,j}=\frac{1}{c_{i,j}+d_{i,j}\eta}$, this minor is of the form
\begin{align*}
M=\frac{p(\eta)}{\prod_{i,j=1}^3(c_{i,j}+d_{i,j}\eta)},
\end{align*}
where $p(\eta)$ is a polynomial in $\eta$ of degree at most $6$. Hence $M$ can vanish for at most six values of $\eta$, which implies the theorem.
\end{proof}

\cref{thm:not_GRS_one_twist} directly implies the existence of non-GRS MDS twisted codes for many field sizes, as appears from the following corollaries.

\begin{corollary}
  \label{cor:existence_non_GRS_star_plus_codes}
   Suppose $(\Fq^*, \cdot)$ has a non-trivial subgroup $G$ such that $|\Fq^* \setminus G| >6$. Then for any $n, k$ with $2 < k < n-2$ and $n \leq |G|$ there exists a non-GRS MDS \startw code. Similarly if $(\Fq,+)$ has a non-trivial subgroup $V$ such that $|\Fq \setminus V|>6$, then for any $n, k$ with $2 < k < n-2$ and $n \leq |V|$ there exists a non-GRS MDS \plustw code.
\end{corollary}

\begin{corollary}\label{cor:existence_non_GRS_CTRS_one_twist}
Let $\Fs \subsetneq \Fq$ with $|\Fq \setminus \Fs|>6$. Let $2 < k < n-2$ and $n\leq s$.
Then, there exists $\eta \in \Fq \setminus \Fs$ such that $\Csingle$ is MDS but not equivalent to a GRS code.
\end{corollary}

\begin{remark}\label{rem:existence_non-GRS_single-twisted}
From $q \geq 13$, non-GRS \startw or \plustw codes of length $\lceil(q+1)/2 \rceil$ is guaranteed by \cref{cor:existence_non_GRS_star_plus_codes}.
If $q=p^m$ is a prime power with $p>3$ and $m>1$, then the restriction $|\Fq \setminus \Fs|>6$ of \cref{cor:existence_non_GRS_CTRS_one_twist} is fulfilled for all $\Fs \subsetneq \Fq$.
\end{remark}

\section{Computer Searches}
\label{sec:Computer_Searches}

In this section, we present exhaustive computer searches for twisted codes over small field sizes.
Since we are most interested in non-GRS codes, we only perform searches for $\eta \neq 0$ and $\min\{k,n-k\}>2$ (cf.~\cref{cor:min_k_n-k_>_2}).
The compututations were carried out using SageMath v7.4 \cite{stein_sagemath_????}.
The full results and the source code can be downloaded from \url{http://jsrn.dk/code-for-articles}.

\subsection{Number of \Startw Codes}

We counted all \startw codes over $\Fq$ for $q \leq 19$ and all \plustw codes over $\Fq$ for $q = 16$ and $q = 49$, i.e., the number of sets $S$ and $\eta$'s that fulfill the conditions of \cref{def:startwisted} respectively \cref{def:plustwisted}.
Moreover, we have determined how many of the resulting codes are inequivalent and which of those are not GRS codes.

As predicted, for odd $q$, there are \startw codes of length up to $\frac{q+1}{2}$ and arbitrary $k$. It also turns out that almost all \startw codes are non-GRS.
In particular, there is exactly 1 \startw $[11,6,3]$ code which is not GRS, even though \cref{cor:min_k_n-k_>_2} did not guarantee this.

\cref{tab:star_twisted_computer_search_example} exemplifies the results for $q=19$.

\begin{table}[h]
\caption{Number of \startw codes over $\F_{19}$ (Total / Inequivalent / Non-GRS).}
\label{tab:star_twisted_computer_search_example}
\centering
\renewcommand{\arraystretch}{1.3}
\begin{tabular}{c|c|c|c|c|c}
$n\backslash k$ & $3$ & $4$ & $5$ & $6$ & $7$ \\
\hline
$6$ & $1974/73/67$ &   &  &   &     \\
$7$ & $1092/67/67$ & $1092/63/63$  &  &   &     \\
$8$ & $405/25/25$ & $405/25/25$  & $405/25/25$ &   &     \\
$9$ & $90/7/7$ & $90/6/6$  & $90/6/6$ & $90/7/7$  &     \\
$10$ & $9/2/2$ & $9/1/1$  & $9/1/1$ & $9/2/2$  & $9/1/1$
\end{tabular}
\end{table}

\subsection{Comparison with Roth--Lempel Codes}

Roth and Lempel \cite{roth_construction_1989} gave a construction of non-GRS MDS codes: given $S \subset \F_q \cup \{\infty\}$ with $n = |S|$ which is \emph{not} a $(k-1)$-sum generator of $(\F_q, +)$, it produces an $[n,k]$ MDS code\footnote{%
  We say that a set $S$ containing $\infty$ is a $k$-sum generator if $S \setminus \{\infty\}$ is a $k$-sum generator.
  Note that for comparison with our codes, we relax the construction of \cite{roth_construction_1989} by allowing $S$ which does not contain $\infty$.
  }.
Roth and Lempel point out, similar to our \cref{def:plustwisted}, that e.g.~subgroups of $\F_q$ will give such non-$(k-1)$-generators.
When $q$ is even, these explicit constructions allow codes in a similar range as \plustw codes.
When $q$ is odd, their construction is much worse; for $q$ an odd prime, they remark in \cite{roth_t-sum_1992} that asymptotically their construction has at most length $\frac q k(1 + o(1))$.

For small $q$ one can exhaustively search for all non-$(k-1)$-sum generators, however, and thereby produce all Roth--Lempel (RL) codes.
We have done this for some parameters with the aim of determining how often \startw or \plustw codes are equivalent to RL codes; especially for the \plustw codes where the possible range of parameters largely coincides.

Our computer searches indicate that the code families are largely independent.
We give three examples:

For $(q, n, k) = (13, 7, 3)$, there are 35 inequivalent RL codes, while there are 2 \startw codes; 1 code is in both sets.
There are no \plustw codes of these parameters, but there are 8 twisted codes with $(t,h) = (1,k-1)$; 2 of these are RL codes.

For $(q, n, k) = (16, 8, 5)$, there are 186 inequivalent RL codes, while there are 9 inequivalent \plustw codes.
These codes are all different.
There are 83 twisted codes in total with $(t,h) = (1,k-1)$, and 10 of these are RL codes.

For $(q, n, k) = (23, 12, 5)$, there are no RL codes, while there is 1 equivalence class of \startw codes.

\subsection{Length $\approx q/2$ Codes with ``Exotic Twists''}

\startw and \plustw codes are explicit subclasses of the cases $(t,h) = (1,0)$ respectively $(t,h) = (1,k-1)$ which allow codes of length $\approx q/2$ for odd respectively even $q$.
A natural question is if similar long MDS codes are possible for twisted codes of other $(t,h)$.

We have no explicit construction, but exhaustive search indicates a resounding 'yes': in fact, for any $q \leq 19$ we verified that for almost \emph{any} choice of $(t, h)$ there is an $[n,k]$ twisted MDS code for $n = \lfloor q/2 \rfloor$ and $3 \leq k \leq n-3$, with the only exceptions being $(t, h) = (1, k-1)$ which fails for $(q,k) = (17, 4)$ and $q = 19$ and any $k$.

The non-existence for $(t, h) = (1, k-1)$ should indeed be expected for high enough $k$ due to \cref{lem:k_sum_gen_plus_MDS} and the asymptotic upper bounds on $M(k, \F_q)$ discussed in \cite{roth_t-sum_1992}.
However, it is surprising and intriguing that for all other $(t,h)$, there seem to be long twisted MDS codes.

\subsection{Counting Twisted MDS Codes}

\cref{tab:Search_Equivalence_Classes} enumerates all MDS twisted codes for $q \leq 13$, the number of equivalence classes as well as non-GRS equivalence classes.
We see that many parameters often lead to equivalent codes: e.g. for $(q,n,k)=(13,7,4)$, each equivalence class is obtainable by $\approx\!\!1200$ parameters on average.
This might be due to algebraic symmetries in the parameter choices, which is an interesting question to investigate.
However, most MDS twisted codes are non-GRS.
Note that in most of the cases, we can construct codes of length $q-1$ for $k \in \{3,n-3\}$.
Apart from Glynn's code and its dual \cite{glynn1986non}, we find only a single new code of length $q$: $(q,n,k)=(8,8,5)$ and constructible as e.g.~$\TRS{\alphaVec,1,4,1}{5}{}$ with $\alphaVec = \F_8 \cup \{ \infty \} \setminus \{ 1 \}$.
This code is also equivalent to a Roth--Lempel code.

{
\newcommand{\Count}[3]{#1$/$#2$/$#3}
\newcommand{\ZeroCount}{}

\begin{table}[b]
\caption{Counting MDS Twisted Codes (Total/Inequivalent/Non-GRS. Blank = 0/0/0)}
\label{tab:Search_Equivalence_Classes}
\centering
\renewcommand{\arraystretch}{1.2}
\setlength{\tabcolsep}{1pt}

\tiny
\begin{tabular}{@{}c|c||c|c|c|c|c|c|c|c@{}}
$q$ & $n$ & $k=3$                   & $4$                   & $5$                     & $6$                  & $7$                & $8$               & $9$ \\
\hline \hline
7
    & 6   & \Count{38}{3}{2}        &                       &                         &                      &                    &                   &     \\
\hline
8
    & 6   & \Count{406}{5}{4}       &                       &                         &                      &                    &                   &     \\
    & 7   & \ZeroCount              & \Count{63}{2}{1}      &                         &                      &                    &                   &     \\
    & 8   & \ZeroCount              & \ZeroCount            & \Count{14}{2}{1}        &                      &                    &                   &     \\
\hline
9
    & 6   & \Count{2374}{7}{5}      &                       &                         &                      &                    &                   &     \\
    & 7   & \Count{216}{3}{3}       & \Count{332}{3}{2}     &                         &                      &                    &                   &     \\
    & 8   & \Count{4}{1}{1}         & \Count{40}{1}{1}      & \Count{36}{1}{1}        &                      &                    &                   &     \\
    & 9   & \ZeroCount              & \Count{4}{1}{1}       & \Count{4}{1}{1}         & \ZeroCount           &                    &                   &     \\
\hline
11
    & 6   & \Count{32518}{15}{11}   &                       &                         &                      &                    &                   &     \\
    & 7   & \Count{6286}{21}{19}    & \Count{8554}{20}{18}  &                         &                      &                    &                   &     \\
    & 8   & \Count{585}{15}{15}     & \Count{160}{7}{6}     & \Count{960}{9}{7}       &                      &                    &                   &     \\
    & 9   & \Count{40}{3}{3}        & \ZeroCount            & \Count{20}{1}{0}        & \Count{135}{1}{0}    &                    &                   &     \\
    & 10  & \Count{2}{1}{1}         & \ZeroCount            & \ZeroCount              & \ZeroCount           & \Count{22}{2}{1}   &                   &     \\
\hline
13
    & 6   & \Count{216722}{26}{21}  &                       &                         &                      &                    &                   &     \\
    & 7   & \Count{71618}{80}{75}   & \Count{98430}{80}{75} &                         &                      &                    &                   &     \\
    & 8   & \Count{11164}{165}{160} & \Count{5176}{98}{93}  & \Count{26916}{139}{134} &                      &                    &                   &     \\
    & 9   & \Count{1110}{32}{31}    & \Count{41}{4}{3}      & \Count{381}{8}{5}       & \Count{5424}{24}{21} &                    &                   &     \\
    & 10  & \Count{138}{4}{4}       & \ZeroCount            & \ZeroCount              & \Count{93}{3}{0}     & \Count{1167}{4}{1} &                   &     \\
    & 11  & \Count{24}{1}{1}        & \ZeroCount            & \ZeroCount              & \ZeroCount           & \ZeroCount         & \Count{254}{1}{0} &     \\
    & 12  & \Count{2}{1}{1}         & \ZeroCount            & \ZeroCount              & \ZeroCount           & \ZeroCount         & \ZeroCount        & \Count{26}{2}{1}
\end{tabular}
\end{table}
}

\section{Decoding of Twisted Codes}

A simple decoding paradigm is possible for twisted codes:
guess the value of the hook coefficient $a_h$, and then apply an $[n, k]$-RS decoding algorithm on $r - \ev{\eta a_h x^{k-1+t}}{\alphaVec}$, where $r$ is the received word.
If the twisted code is over $\F_q$, this will apply the RS decoder $q$ times.
This works with errors, erasures, soft-decision, list-decoding, etc.
In particular, we have (cf.~\cite{justesen_complexity_1976}):

\begin{theorem}
  An $[n,k]$ twisted code over $\F_q$ can be decoded up to half the minimum distance in complexity $\Osoft(qn)$.
\end{theorem}

Note that even if the twisted code is not MDS, this approach still gives a list-decoder up to $(n-k+1)/2$: collect the codewords obtained from each guess of $a_h$, and the correct codeword is on the resulting list if the number of errors is at most the decoding radius of the RS decoder.

\section{Conclusion}

We have introduced twisted Reed--Solomon codes, a new class of $\Fq$-linear codes, and demonstrated that for $q \ge 7$ the class contains many new non-GRS MDS codes.
We singled out two explicit subclasses, \startw and \plustw codes, with which we can construct MDS codes of length at least $q/2$ for any field size $q$, and that for most field sizes, most of these codes will be non-GRS.
Using computer searches we demonstrated that there seems to be many other and longer MDS, non-GRS twisted RS codes.

\bibliographystyle{IEEEtran}
\bibliography{main}

\end{document}